\documentclass[conference,a4paper]{IEEEtran}

\usepackage{cite}
\usepackage{amsmath}
\usepackage{amsfonts}
\usepackage{pifont}
\usepackage{amssymb}
\usepackage{amsthm}
\usepackage{tikz}
\usepackage{graphicx}
\usepackage[caption=false,font=footnotesize]{subfig}
\usepackage{multirow}
\usepackage{url}
\usepackage{algorithmic}
\usepackage{enumerate}
\usepackage{makecell}
\theoremstyle{plain}

\newtheorem{theorem}{Theorem}
\newtheorem{lemma}[theorem]{Lemma}
\newtheorem{proposition}[theorem]{Proposition}
\newtheorem{definition}[theorem]{Definition}
\newtheorem{corollary}[theorem]{Corollary}
\theoremstyle{definition}
\newtheorem{example}{Example}
\newtheorem{construction}{Construction}

\DeclareMathOperator*{\Min}{\text{\upshape Min}}
\newcommand{\supp}{\text{Supp}}
\newcommand{\wt}{\text{wt}}
\renewcommand{\vec}[1]{\boldsymbol{#1}}
\newcommand{\spn}{\text{Span}}

\begin{document}
\title{Bounds and Constructions for Linear Locally Repairable Codes over Binary Fields}

\author{\IEEEauthorblockN{Anyu~Wang\textsuperscript{*},~Zhifang~Zhang\textsuperscript{\dag},~and~Dongdai Lin\textsuperscript{*}}
\IEEEauthorblockA{\fontsize{9.8}{12}\selectfont\textsuperscript{*}State Key Laboratory of Information Security, Institute of Information Engineering, CAS, Beijing, China\\
\textsuperscript{\dag}KLMM, NCMIS, Academy of Mathematics and Systems Science, University of Chinese Academy of Sciences\\
Emails: wanganyu@iie.ac.cn, zfz@amss.ac.cn, ddlin@iie.ac.cn}
}

\maketitle
\thispagestyle{empty}

\begin{abstract} For binary $[n,k,d]$ linear locally repairable codes (LRCs), two new upper bounds on $k$ are derived. The first one applies to LRCs with disjoint local repair groups, for general values of $n,d$ and locality $r$, containing some previously known bounds as special cases. The second one is based on solving an optimization problem and applies to LRCs with arbitrary structure of local repair groups. Particularly, an explicit bound is derived from the second bound when $d\geq 5$. A specific comparison shows this explicit bound outperforms the Cadambe-Mazumdar bound for $5\leq d\leq 8$ and large values of $n$. Moreover, a construction of binary linear LRCs with $d\geq6$ attaining our second bound is provided.

\end{abstract}

\section{Introduction}
Recently, locally repairable codes (LRCs) have attracted a lot of attention due to their applications in distributed storage systems.
An $[n,k,d]$ linear code is called an LRC with locality $r$ if the value at each coordinate can be recovered by accessing at most $r$ other coordinates.
An LRC with small locality $r$ is preferred in practice as it greatly reduces the disk I/O complexity in repairing node failures. Meantime,  large values of $k$ and  $d$ are also desirable to ensure high level of storage efficiency and global fault tolerance ability respectively.
Much work has been done toward exploring the relationship between the parameters $n,k,d,r$.
The first trade-off is derived in \cite{gopalan2012locality}, i.e.,
\begin{equation}\label{eqSingleton}
  d \le n - k - \bigl \lceil \frac{k}{r} \bigr\rceil +2,
\end{equation}
which is also known as the Singleton-like bound for LRCs.
Then various methods are developed to construct LRCs attaining the bound \eqref{eqSingleton}, e.g., \cite{tamo2013optimal,silberstein2013optimal,papailiopoulos2012locally,tamo2014family}.
Tightness of the singleton-like bound is discussed in \cite{song2014optimal,hao2016bounds}, and some
improved  bounds are  derived in \cite{prakash2014codes,wang2015integer}.

It can be seen the bound \eqref{eqSingleton} does not care about the field size. However, in practice LRCs over small finite fields, especially those over binary fields, are preferred due to their convenience in implementation.
The first trade-off taking into consideration the field size is  derived by Cadambe and Mazumdar \cite{cadambe2015bounds}, i.e.,
\begin{equation}\label{eqCM}
  k \le \Min_{t \in \mathbb{Z}^+} \bigl[ tr + k^{(q)}_{opt}(n-(r+1)t,d) \bigr],
\end{equation}
where $k^{(q)}_{opt}(n,d)$ is the largest possible dimension of an $[n,k,d]$ linear code over $\mathbb{F}_q$.
This trade-off is usually called the C-M bound, and is proved achievable by the binary Simplex codes \cite{cadambe2015bounds}.
Another class of binary LRCs with $r=2,3$ constructed via anticodes in \cite{silberstein2015optimalbinary} also meets \eqref{eqCM} with equality.
Although these codes are optimal with respect to the C-M bound, their code length $n$ increases exponentially as the dimension $k$ grows, implying poor performance in information rate. Alternatively, cyclic codes provides more desirable candidates for LRCs over small fields.
In \cite{goparaju2014binary}, binary LRCs with $r=2$ and $d=2,6,10$ are constructed from primitive cyclic codes.
These codes do not attain the C-M bound, but are shown to be optimal under a structural assumption that the codeword coordinates are divided into disjoint local repair groups.
The same method is adopted in \cite{zeh2015optimal} to generate binary LRCs with $r=2, d=10$ from nonprimitive cyclic codes.
In \cite{tamo2015cyclic}, BCH-type binary LRCs are constructed as the subfield subcodes of optimal Reed-Solomon-Type LRCs.
Besides, some other approaches for constructing LRCs that attain the Singleton-like bound \eqref{eqSingleton} over small fields are also developed in \cite{shahabinejad2016class,hao2016some}.

Recently, two upper bounds taking the field size into account are derived in \cite{hu2016combinatorial} for $(r,\delta)$-LRCs.
Since $(r,\delta)$-LRCs \cite{prakash2012optimal} contain LRCs as the special case of $\delta=2$, these two bounds apply to LRCs as well.
For $\delta =2$, the first bound is equivalent to the Singleton-like bound \eqref{eqSingleton}, while the second one is a linear programming bound for LRCs with disjoint repair groups.
On the other hand, asymptotic bounds on the parameters of LRCs are studied in \cite{tamo2016bounds,agarwal2016bounds}.

Overall, most of the bounds derived so far for LRCs over particular finite fields either depend on undetermined parameters in coding theory, e.g., the C-M bound, or rely on solving optimization problems under concrete code parameters, e.g., the LP bound in \cite{hu2016combinatorial}.
And the constructions of binary LRCs with good parameters mostly restrict to specific values of $d$ and $r$. Much work remains undone for LRCs over particular finite fields. In this work, we focus on linear LRCs over binary fields.

\subsection{Main Idea and Contribution}
For any $[n,k,d]$ binary linear LRC $\mathcal{C}$, our main idea for deriving upper bounds on $k$ is to consider a related sphere packing problem in a particular space, namely, the $\mathcal{L}$-space.
Specifically, an $\mathcal{L}$-space of $\mathcal{C}$ is defined to be the dual of the linear space spanned by a minimum set of local parity checks with overall supports covering all coordinates.
Actually, the $\mathcal{L}$-space can be viewed as an LRC which contains $\mathcal{C}$ as a subcode.
Then by applying the sphere packing bound in the $\mathcal{L}$-space, two upper bounds for $\mathcal{C}$ are derived.

Firstly, assuming the code $\mathcal{C}$ has disjoint local repair groups, we get an explicit bound (i.e. Corollary \ref{corDisjoint}) on $k$ for general values of $n,d,r$. Note that for $r=2$ and special forms of $n,d$, upper bounds were also derived in \cite{goparaju2014binary,zeh2015optimal}. It turns out our bound contains their results as special cases.

Secondly, for linear binary LRCs with arbitrary local repair groups, we derive an upper bound (i.e. Theorem \ref{thmSPB}) on $k$ based on solving an optimization problem. Although it is generally difficult to solve this optimization problem, simplification can be done for $d\geq 5$, and thus an explicit upper bound (i.e. Theorem \ref{thm5}) is derived. Through a specific comparison, we show this bound can outperform the C-M bound for $5\leq d\leq 8$ and large values of $n$. Moreover, a class of binary linear LRCs with $d\geq6$ attaining this explicit bound is constructed.

\subsection{Organization}
Section II defines the $\mathcal{L}$-space, and derives our first  bound (i.e. Corollary \ref{corDisjoint}). Section III presents our second bound (Theorem \ref{thm5}). Section IV gives the construction attaining our second bound. Section V concludes the paper.

\section{The $\mathcal{L}$-Space for LRCs}
For any vector $\vec{v} = (v_1,\dots,v_n) \in \mathbb{F}_2^n$, let $\supp{(\vec{v})} = \{i \in [n]: v_i \neq 0\}$ and $\wt{(\vec{v})} = |\supp{(\vec{v})}|$, where $[n]=\{1, 2, \dots, n\}$.
For any two vectors $\vec{u}, \vec{v} \in \mathbb{F}_2^n$, $\text{dist}(\vec{u},\vec{v})$ denotes the hamming distance of $\vec{u}$ and $\vec{v}$.
Denote by $\spn_{2}(\vec{u}_1,\dots,\vec{u}_l)$  the linear space spanned by a set of vectors $\{\vec{u}_1,\dots,\vec{u}_l\}$ over $\mathbb{F}_2$.

Let $\mathcal{C}$ be an $[n,k,d]$ binary linear LRC with locality $r$.
Then for each coordinate $i \in [n]$, there is a local parity check $\vec{h}_i \in \mathcal{C}^\bot$ such that $i \in \supp(\vec{h}_i)$ and $\wt(\vec{h}_i) \le r+1$.
Note by local parity checks we mean the codewords in the dual code $\mathcal{C}^\bot$ with weight at most $r+1$.

\begin{definition}
Let $H \subseteq \{\vec{h}_{1},\dots,\vec{h}_{n}\}$ be a set of local parity checks of $\mathcal{C}$ such that $\bigcup_{\vec{h} \in H} \supp{(\vec{h})} = [n]$, and $\bigcup_{\vec{h} \in H'} \supp{(\vec{h})} \neq [n]$ for any $H' \subsetneq H$. We call $H$ an $\mathcal{L}$-cover of $\mathcal{C}$. Denote $\mathcal{H}=\spn_{2}(H)$, then the dual space of $\mathcal{H}$, i.e., $\mathcal{V}=\{\vec{v}\in \mathbb{F}_2^n\mid \vec{v}\cdot\vec{h}=0, ~\forall \vec{h}\in \mathcal{H}\}$, is called an $\mathcal{L}$-space of $\mathcal{C}$.
\end{definition}

Obviously, an $\mathcal{L}$-cover $H$ contains the minimum number of local parity checks guaranteeing the locality $r$ for all coordinates. $H$ needs not be unique, neither does the $\mathcal{L}$-space $\mathcal{V}$. Our proofs in this paper only depend on their existence which is ensured by the definition of LRCs. Since $H$ only contains partial parity checks of $\mathcal{C}$, it follows that $\mathcal{V}$ also defines an LRC containing $\mathcal{C}$ as a subcode. Investigating the structure of $\mathcal{V}$ may help us to study the code $\mathcal{C}$. In the following, we use the sphere-packing bound in the $\mathcal{L}$-space $\mathcal{V}$, and obtain a connection between $k,d$ and $\mathcal{V}$.

\begin{proposition}\label{propLSB}
For an $[n,k,d]$ binary LRC $\mathcal{C}$ with locality $r$, it holds
\begin{equation}\label{eqLSB}
    k \le \dim(\mathcal{V}) - \log_2 \bigl(B_{\mathcal{V}}(\bigl\lfloor \frac{d-1}{2} \bigr\rfloor) \bigr),
\end{equation}
where $B_{\mathcal{V}}(\bigl\lfloor \frac{d-1}{2} \bigr\rfloor) = \left| \{ \vec{v} \in \mathcal{V} : \wt(\vec{v}) \le \bigl\lfloor \frac{d-1}{2} \bigr\rfloor\} \right|$, and $\mathcal{V}$ is an $\mathcal{L}$-space of $\mathcal{C}$.
\end{proposition}
\begin{proof}
For any codeword $\vec{c} \in \mathcal{C}$, consider the ball of radius $\left\lfloor \frac{d-1}{2} \right\rfloor$ around $\vec{c}$  in $\mathcal{V}$.
Since $\mathcal{C}$ has minimum distance $d$, then these balls are non-overlapping.
It follows that
\begin{equation}\label{eqSp0}
  \sum_{\vec{c} \in \mathcal{C}} B_{\mathcal{V}}(\vec{c},\bigl\lfloor \frac{d-1}{2} \bigr\rfloor) \le |\mathcal{V}|,
\end{equation}
where $B_{\mathcal{V}}(\vec{c}, \bigl\lfloor \frac{d-1}{2} \bigr\rfloor) = \left| \{\vec{v} \in \mathcal{V}: \text{dist}(\vec{v},\vec{c}) \le \left\lfloor \frac{d-1}{2} \right\rfloor \} \right|$.
Note that $\mathcal{C} \subseteq \mathcal{V}$, so we have $B_{\mathcal{V}}(\vec{c},\bigl\lfloor \frac{d-1}{2} \bigr\rfloor)= B_{\mathcal{V}}(\bigl\lfloor \frac{d-1}{2} \bigr\rfloor), \forall \vec{c} \in \mathcal{C}$.
Therefore \eqref{eqSp0} can be written as
\begin{equation}\label{eqSp1}
  |\mathcal{C}| \cdot B_{\mathcal{V}}( \bigl\lfloor \frac{d-1}{2} \bigr\rfloor) \le |\mathcal{V}|.
\end{equation}
Since $\log_2 |\mathcal{C}| = k$ and $ \log_2 |\mathcal{V}| = \dim(\mathcal{V})$, then the Theorem follows directly from \eqref{eqSp1}.
\end{proof}

The right hand side of \eqref{eqLSB} depends on the locality space $\mathcal{V}$,
so explicit bound can be derived from \eqref{eqLSB} if $\mathcal{V}$ is known.
In the following, we apply Proposition \ref{propLSB} to a special class of binary linear LRCs which has a clear $\mathcal{L}$-space.

\subsection{Bound for LRCs with Disjoint Local Repair Groups}
Assume $\mathcal{C}$ has local parity checks $\vec{h}_{i_1},\vec{h}_{i_2},\dots,\vec{h}_{i_l} \in \mathcal{C}^\bot$ satisfying $\cup_{j=1}^l \supp(\vec{h}_{i_j}) = [n]$, $\wt(\vec{h}_{i_j}) = r+1$ and $\supp(\vec{h}_{i_j}) \cap \supp(\vec{h}_{i_{j'}}) = \emptyset$ for $1 \le j \ne j' \le l$. Obviously, $r+1\mid n$ and $l=\frac{n}{r+1}$.
Such an LRC is usually said to have  \emph{disjoint local repair groups}, which is widely adopted in constructions of LRCs, e.g., \cite{papailiopoulos2012locally,tamo2014family,silberstein2013optimal,tamo2013optimal}.
Under this assumption, the structure of the $\mathcal{L}$-space $\mathcal{V}$ becomes quite simple. Then based on Proposition \ref{propLSB}, we derive the following upper bound.

\begin{corollary}\label{corDisjoint}
For any $[n,k,d]$ binary LRC $\mathcal{C}$ with locality $r$ that has disjoint local repair groups, it holds
\begin{equation}\label{eqSpbPre}
k \le \frac{rn}{r+1} - \log_2 \bigl( \sum_{0 \le i_1 +\dots+i_l \le \lfloor \frac{d-1}{4}\rfloor  } \prod_{j=1}^l \binom{r+1}{2 i_j} \bigr)\;.
\end{equation}
\end{corollary}

\begin{proof}
  Note that $H=\{\vec{h}_{i_1},\vec{h}_{i_2},\dots,\vec{h}_{i_l}\}$ is an $\mathcal{L}$-cover of $\mathcal{C}$, then $\mathcal{V} = \spn_{2}(H)^\bot$ is an $\mathcal{L}$-space of $\mathcal{C}$.
  By Proposition~\ref{propLSB}, it suffices to determine $\dim(\mathcal{V})$ and $B_{\mathcal{V}}(\bigl\lfloor \frac{d-1}{2} \bigr\rfloor)$.
  Clearly it has $\dim(\mathcal{V}) =\frac{rn}{r+1}$.
  Note that the linear space $\spn_{2}(H)$ has weight enumerator polynomial $W_{H}(x,y) = (x^{r+1} + y^{r+1})^l$.
  Then by the MacWilliams equality, (see e.g., \cite{macwilliams1977theory}), the weight enumerator polynomial of $\mathcal{V}$ is
  \begin{align*}
    W_{\mathcal{V}}(x,y) & = \frac{1}{2^l} W_{H}(x+y,x-y) \\
    & = \frac{1}{2^l} ((x+y)^{r+1} + (x-y)^{r+1})^l \\\
    & = \bigl( \sum_{i \ge 0} \binom{r+1}{2 i}x^{r+1-2i}y^{2i} \bigr)^l \\
    & = \sum_{0\le u \le \frac{n}{2}} A_{u} x^{n-2u} y^{2u},
  \end{align*}
  where $A_{u} = \sum\limits_{i_1+\dots+i_l = u} \prod_{j=1}^l \binom{r+1}{2 i_j}$.
  Thus we have
  \begin{align*}
    B_{\mathcal{V}}(\bigl\lfloor \frac{d-1}{2} \bigr\rfloor) & = A_0 + \dots + A_{\lfloor\frac{d-1}{4}\rfloor} \\
    & = \sum_{0 \le i_1 +\dots+i_l \le \lfloor \frac{d-1}{4}\rfloor  } \prod_{j=1}^l \binom{r+1}{2 i_j} \bigr),
  \end{align*}
  and \eqref{eqSpbPre} follows directly.
\end{proof}

The sphere packing approach was also used in \cite{goparaju2014binary,zeh2015optimal} to derive upper bounds on $k$ for binary linear LRCs with disjoint local repair groups.
However, their approach only applies to the case of $r=2$ because it relies on a map from binary linear LRCs with $r=2$ to additive $\mathbb{F}_4$-codes. Our bound works for general values of $n,d,r$, especially containing the bounds in
\cite{goparaju2014binary,zeh2015optimal} as special cases.
For example, suppose $n=2^m-1, d=6$ and $r=2$, then Corollary \ref{corDisjoint} implies that
\begin{align*}
  k & \le \frac{2n}{3} - \log_2 \bigl(\binom{r+1}{0} + l \binom{r+1}{2} \bigr) \\
  & = \frac{2n}{3} - \log_2 (1+n) \\
  &= \frac{2}{3} (2^m-1) -m,
\end{align*}
which coincides with the Theorem 1 in \cite{goparaju2014binary}.

Another bound for LRCs with disjoint repair groups is the LP bound derived in \cite{hu2016combinatorial}.
Table \ref{Table0} lists a comparison of the bound \eqref{eqPreBnd}, the LP bound in \cite{hu2016combinatorial} and the C-M bound \eqref{eqCM} for $3 \le r \le 10, \frac{n}{r+1}=3, d=5$.
From the table we can see the bound \eqref{eqSpbPre} is slightly weaker than the LP bound but tighter than the C-M bound \eqref{eqCM}.
Nevertheless, the bound \eqref{eqSpbPre} has an explicit form and can be more easily implemented than the other two bounds.

\renewcommand{\arraystretch}{1.2}
\begin{table}[!htb]
\centering
\begin{tabular}[b]{|c|c|c|c|c|c|c|c|c|c|}
\hline
$r$ & $3$ & $4$ & $5$ & $6$ & $7$ & $8$ & $9$ & $10$ \\ \hline
Our bound \eqref{eqSpbPre} & $4$ & $7$ & $9$ & $12$ & $14$ & $17$ & $19$ & $22$ \\
The C-M bound \eqref{eqCM} & $5$ & $7$ & $10$ & $13$ & $15$ & $18$ & $21$ & $23$ \\
The LP bound \cite{hu2016combinatorial}  & $4$ & $6$ & $9$ & $11$ & $14$ & $17$ & $19$ & $22$ \\
\hline
\end{tabular}
\caption{}
\label{Table0}
\end{table}
\renewcommand{\arraystretch}{1}

\section{New Upper Bound for Binary Linear LRCs}\label{secNB}
In this section we will remove the assumption of disjoint local repair groups, and derive parameter bounds for linear binary LRCs with arbitrary local repair groups.
Suppose $\mathcal{C}$ is an $[n,k,d]$ binary linear LRC with locality $r$.
Let $H = \{\vec{h}_{i_1},\dots,\vec{h}_{i_l}\} \subseteq \mathcal{C}^\bot$ be an $\mathcal{L}$-cover of $\mathcal{C}$.
By shortening at the coordinates that appear more than once in the supports of $\vec{h}_{i_1},\dots,\vec{h}_{i_l}$, we can derive from $\mathcal{C}$ a shortened code $\mathcal{C'}$ which has disjoint local repair groups. Thus the problem is reduced to that we discussed in last section.

Specifically, define $L \in \mathbb{F}_2^{l\times n}$ to be a matrix whose rows are the $l$  local parity checks in $H$.
Denote by $N$ the number of columns in $L$ that have weight $1$. Let $L'$ be the matrix obtained from $L$ by deleting the $n-N$ columns of $L$ that have weight greater than $1$. Then taking $L'$ as the parity check matrix defines a binary code $\mathcal{C}'$. It can be proved $\mathcal{C}'$ has the following properties.

\begin{lemma}\label{lemShorten}
  The shortened code $\mathcal{C}'$ is an $[N,K,D]$ binary linear LRC satisfying
  \begin{itemize}
    \item[(i)] $n \ge N \ge 2n- l(r+1)$, $K \ge  N- (n-k) $, $D \ge d$;
    \item[(ii)] $\mathcal{C}'$ has an $\mathcal{L}$-cover  $H'=\{\vec{h}'_{i_1}, \dots,\vec{h}'_{i_l}\}$  such that $1 \le \wt(\vec{h}'_{i_j}) \le  r+1$ and $\supp(\vec{h}'_{i_j})\cap \supp(\vec{h}'_{i_{j'}}) = \emptyset$ for all $1 \le j \neq j' \le n$.
  \end{itemize}
\end{lemma}
\begin{proof}
  Since the shortening operation neither increases the redundancy nor decreases the minimum distance, (see e.g., \cite{macwilliams1977theory}), then it has $K \ge  N- (n-k) $ and $D \ge d$.
  To show the other statements, we suppose without loss of generality that
  \begin{equation*}
    L = \bigl( L', \;\; L''  \bigr),
  \end{equation*}
  where $L' \in \mathbb{F}_2^{ l \times N}$ consists of the $N$ columns that have weight $1$, and $L'' \in \mathbb{F}_2^{ l \times (n-N)}$ consists of the other $(n-N)$ columns that have weight $\ge 2$.
  By counting the number of $1$'s in $L$, we have
  \begin{align*}
    l(r+1) & \ge \text{ the number of $1$'s in $L$} \\
    & \ge N + 2(n-N).
  \end{align*}
  Thus $2n - l(r+1) \le N \le n$.
  Lastly, denote $\vec{h}'_1, \dots,\vec{h}'_l$ to be the rows of $L'$, then clearly $\{\vec{h}'_1, \dots,\vec{h}'_l\}$ is a set of parity checks of $\mathcal{C}'$ such that $1 \le \wt{(\vec{h}'_i)} \le r+1$.
  Note that each column of $L'$ has exactly one $1$, therefore $\cup_{i=1}^l \supp{(\vec{h}'_i)} = [N]$ and $\supp(\vec{h}'_i)\cap \supp(\vec{h}'_j) = \emptyset$ for all $1 \le i \neq j \le n$, which completes the proof.
\end{proof}

Let $\mathcal{V}' = \spn_{2}(H')^\bot$ be an $\mathcal{L}$-space of $\mathcal{C}'$, and denote $\wt(\vec{h}_{i_j}' ) = r_j+1$ for $j \in [l]$.
Then it follows $\dim(\mathcal{V}')= N-l$, and by a deduction similar to that in Corollary \ref{corDisjoint} it has
\begin{equation*}
  B_{\mathcal{V}'}( \lfloor\frac{D-1}{2} \rfloor) = \sum_{0 \le i_1+\dots+i_l \le \lfloor\frac{D-1}{4} \rfloor} \prod_{j=1}^l \binom{r_j+1}{2 i_j}\;.
\end{equation*}
Applying Proposition \ref{propLSB} to the shortened LRC $\mathcal{C}'$, we get
\begin{equation}\label{eqPreBnd}
  K \le (N-l) - \log_2 \bigl( \sum_{0 \le i_1+\dots+i_l \le \lfloor\frac{D-1}{4} \rfloor} \prod_{j=1}^l \binom{r_j+1}{2 i_j} \bigr).
\end{equation}
Then combining with Lemma \ref{lemShorten}, we get the following theorem.

\begin{theorem}\label{thmSPB}
    For any $[n,k,d]$ binary linear LRC with locality $r$, it holds
  \begin{equation}\label{eqSPB}
    k \le n - \Min_{l,r_1,\dots,r_l} \Big[ l+\log_2 \left( \Phi_l(r_1,\dots,r_l) \right) \Big],
  \end{equation}
  where
  \begin{equation*}
  \Phi_l(r_1,\dots,r_l) \!=\! \sum\limits_{0 \le i_1+\dots+i_l \le \lfloor\frac{d-1}{4} \rfloor} \prod_{j=1}^l \binom{r_j+1}{2 i_j}
  \end{equation*}
  and the `Min' is taken over all integers $l,r_1,\dots,r_l$ such that
  \begin{equation}\label{eqConditions}
    \begin{cases}
      \frac{n}{r+1} \le l \le \frac{2n}{r+2}; \\
       0 \le r_1, \dots, r_l \le r ;\\
        r_1+\dots+r_l = 2n - l(r+2).
    \end{cases}
  \end{equation}
\end{theorem}
\begin{proof}
    From Lemma \ref{lemShorten} it has $K \ge N - (n-k), D \ge d$.
  Then
  \begin{align*}
    k & \le K-N +n \\
    &  \substack{(a) \\ \le}\,\, n-l - \log_2 \bigl( \sum_{0 \le i_1+\dots+i_l \le \lfloor\frac{D-1}{4} \rfloor} \prod_{j=1}^l \binom{r_j+1}{2 i_j} \bigr) \\
    & \substack{(b) \\ \le}\,\, n-l - \log_2 \left( \Phi_l(r_1,\dots,r_l) \right),
  \end{align*}
  where (a) is from \eqref{eqPreBnd} and (b) holds because $D \ge d$.
  Note that the integers $l,r_1,\dots,r_l$ satisfies
  \begin{equation*}
    \begin{cases}
      \frac{n}{r+1} \le l; \\
      0 \le r_1,\dots,r_l \le r; \\
      \sum_{j=1}^l r_j \ge 2n - l (r+2).
    \end{cases}
  \end{equation*}
  There are two cases.

  {\it Case 1: $l > \frac{2n}{r+2}$.} On the one hand, we have $k\le n - l < n - \frac{2n}{r+2}$.
  On the other hand, with the restriction \eqref{eqConditions}, it has
  \begin{align*}
    &\phantom{{}={}} \Min_{l,r_1,\dots,r_l} \Big[ l+\log_2 \left( \Phi_l(r_1,\dots,r_l) \right)\Big] \\
  &\le\frac{2n}{r+2}+\log_2 \left( \Phi_{\frac{2n}{r+2}}(0,\dots,0) \right) \\
  &= \frac{2n}{r+2}.
  \end{align*}
  So inequality \eqref{eqSPB} holds.

  {\it Case 2: $l \le \frac{2n}{r+2}$.} In this case it has
  \begin{equation}\label{eqConditionsPre}
    \begin{cases}
      \frac{n}{r+1} \le l \le\frac{2n}{r+2} ; \\
      0 \le r_1,\dots,r_l \le r; \\
      \sum_{j=1}^l r_j \ge 2n - l (r+2).
    \end{cases}.
  \end{equation}
  So we have
  \begin{equation*}
    k \le n - \Min_{l,r_1,\dots,r_l} \Big[ l+\log_2 \left( \Phi_l(r_1,\dots,r_l) \right) \Big],
  \end{equation*}
  where the `Min' is taken over \eqref{eqConditionsPre}.
  Note that $2n-l(r+2) \ge 0$, so a necessary condition for optimizing \eqref{eqSPB} is that $\sum_{j=1}^l r_j = 2n - l (r+2)$.
  Then the optimization can be restricted to the condition \eqref{eqConditions}, and thus the theorem holds.
\end{proof}

For any given $n,d,r$, Theorem \ref{thmSPB} gives an upper bound on the dimension $k$ based on solving an optimization problem.
However, solving the optimization problem is very difficult in general since the objective function in \eqref{eqSPB} is nonlinear.
Nevertheless, it is still possible to simplify the bound \eqref{eqSPB} in some special cases. Next, we will derive an explicit upper bound  from Theorem \ref{thmSPB} for $d \ge 5$ .

\subsection{Explicit Bound for $d\ge 5$}
\begin{theorem}\label{thm5}
  For any $[n,k,d]$ binary linear LRC with locality $r$ such that $d \ge 5$ and $2 \le r \le \frac{n}{2}-2$, it holds
  \begin{equation}\label{eqBnd5}
    k \le \frac{rn}{r+1} - \min\{ \log_2(1+\frac{rn}{2}) ,\frac{rn}{(r+1)(r+2)} \}.
  \end{equation}
\end{theorem}
\begin{proof}
  When $d \ge 5$, it has $\bigl \lfloor \frac{d-1}{4} \bigr \rfloor \ge 1$ and therefore
  \begin{align*}
    \Phi_l(r_1,\dots,r_l)   & \ge  \sum_{0 \le i_1+\dots+i_l \le 1} \prod_{j=1}^l \binom{r_j+1}{2 i_j} \\
    & =     1+ \binom{r_1+1}{2} + \dots + \binom{r_l+1}{2}.
  \end{align*}
Note that $\binom{x+1}{2} = \frac{1}{2} x(x+1)$ is a convex real-valued function and it is required in \eqref{eqConditions} that $r_1+\dots+r_l = 2n-l(r+2)$, so
  \begin{align*}
    \Phi_l(r_1,\dots,r_l) & \ge 1+ l\binom{\frac{1}{l}\sum_{j=1}^l(r_j+1)}{2} \\
    & = 1 + \frac{1}{2l} \left(2n-l(r+1) \right) \left( 2n-l(r+2) \right)).
  \end{align*}
  It follows from Theorem \ref{thmSPB} that
  \begin{align*}
    k & \le n - \Min_{l,r_1,\dots,r_l} \Big[ l+\log_2 ( \Phi_l(r_1,\dots,r_l) ) \Big] \\
    & \le n - \Min_{l} \Big[ l \! + \! \log_2 \bigl( 1 \! + \!  \frac{(2n-l(r+1) ) ( 2n-l(r+2) )}{2l} \bigr)  \Big],
  \end{align*}
  where the integer $l$ satisfies $\frac{n}{r+1} \le l \le \frac{2n}{r+2}$ according to \eqref{eqConditions}.
  Let
  $$f(l) = l+\log_2 ( 1+ \frac{1}{2l} (2n-l(r+1) ) ( 2n-l(r+2) ))$$
   be a function defined for integers $l \in [\frac{n}{r+1},\frac{2n}{r+2}]$.
   We claim
  \begin{equation*}
    {f(l)} \ge \frac{n}{r+1}+\min \{\log_2(1+\frac{rn}{2}),\frac{rn}{(r+1)(r+2)}\}
  \end{equation*}
  for $2 \le r \le \frac{n}{2} - 2$, then the theorem follows directly.

Firstly, we show that $f''(l) \le 0$ for $2 \le r \le \frac{n}{2} - 2$.
Note that
\begin{equation*}
  f''(l) \!=\! \frac{80n^4 \!-\! (l^2(r^2 \!+\! 3r \!+\! 2) \!-\! 8n^2)^2 \!-\! 16 l (n^3(3 \!+\! 2r) \!-\! n^2)}{l^2 (4 n^2 \!+\! l^2(2 + 3r + r^2) + l(2 - 2n(3 + 2r)))^2 \ln2},
\end{equation*}
then it suffices to prove $g(l) \le 0$ for $2 \le r \le \frac{n}{2} - 2$, where
$g(l) = 80n^4 \!-\! (l^2(r^2 \!+\! 3r \!+\! 2) \!-\! 8n^2)^2 \!-\! 16 l (n^3(3 \!+\! 2r) \!-\! n^2)$.
Since $g'''(l) = -24l(r^2+3r+2)^2 <0$, it has $g'(l)$ is a concave function with $g'(\frac{n}{r+1}) = \frac{4n^2 (4r+4-nr^2)}{r+1} <0$ and $g'(\frac{2n}{r+2}) = \frac{16n^2(2+r+nr)}{r+2} >0$.
It follows that
\begin{align*}
  g(l) & \le \max\{g(\frac{n}{r+1}),g(\frac{2n}{r+2})\} \\
  & =\max\{\frac{n^3(16(r+1) \!-\! n(r+2)^2)}{(r+1)^2},\frac{16n^3(2(r+2) \!-\! n)}{(r+2)^2}\} \\
  & \le 0
\end{align*}
for all $2 \le r \le \frac{n}{2} - 2$.

According to $f''(l) \le 0$, $f(l)$ is a concave function, then we have
\begin{align*}
  {f(l)} & \ge \min \{f(\frac{n}{r+1}),f(\frac{2n}{r+2}) \} \\
  & = \min \{\frac{n}{r+1}+\log_2(1+\frac{rn}{2}),\frac{2n}{r+2}\} \\
  & = \frac{n}{r+1}+\min \{\log_2(1+\frac{rn}{2}),\frac{rn}{(r+1)(r+2)}\},
\end{align*}
and therefore the theorem follows.
\end{proof}

Since we focus on linear codes, then $k$ is actually upper bounded by the largest integer no more than the right hand side of the inequality.
For sufficiently large $n$, it always holds
$\log_2(1+\frac{rn}{2}) < \frac{rn}{(r+1)(r+2)} $. More specifically,
this inequality holds whenever $n \ge 5 (r+1)(r+2)$. Therefore, the bound \eqref{eqBnd5} can be further simplified as $k \le \frac{rn}{r+1} - \bigl\lceil \log_2(1+\frac{rn}{2}) \bigr\rceil.$

Next, we give a comparison between the bound \eqref{eqBnd5} and the C-M bound \eqref{eqCM}, where the upper bound on $k^{(2)}_{opt}(n,d)$ is computed by using SageMath \cite{sage} and the web database \cite{grassl2007codetables}.

For $5 \le d \le 8$, according to our computation, the bound \eqref{eqBnd5} can always outperform the C-M bound for large values of $n$. Specifically,
Fig.~\ref{FigD5} and Fig.~\ref{FigD8} display comparisons of the two bounds for $r=3,d= 5,10 \le n \le 60$ and $r=2,d=8, 60 \le n \le 110$ respectively.
Moreover in Table \ref{Table1}, based on a detailed calculation of the two bounds for  $2 \le r \le 5$ and $2 \le n \le 250$, we list the tipping points of $n$'s that the bound \eqref{eqBnd5} is tighter than  the C-M bound thereafter.

\renewcommand{\arraystretch}{1.2}
\begin{table}[!htb]
\centering
\begin{tabular}[b]{|c|c|c|c|c|c|}
\hline
  & $r=2$ & $r=3$ & $r=4$ & $r=5$  \\ \hline
$d=5$ & $n \ge 19$ & $n \ge 26$ & $n \ge 34$ & $n \ge 43$  \\ \hline
$d=6$ & $n \ge 23$ & $n \ge 31$ & $n \ge 40$ & $n \ge 50$ \\ \hline
$d=7$ & $n \ge 31$ & $n \ge 41$ & $n \ge 84$ & $n \ge 125$  \\ \hline
$d=8$ & $n \ge 50$ & $n \ge 70$ & $n \ge 145$ & $n \ge 228$  \\
\hline
\end{tabular}
\caption{}
\label{Table1}
\end{table}
\renewcommand{\arraystretch}{1}

For $d \ge 9$, it can be checked that the bound \eqref{eqBnd5} is inferior to the C-M bound.
A reason causing this disadvantage is that in this case it has $\bigl \lfloor \frac{d-1}{4} \bigr \rfloor \ge2$ while the bound \eqref{eqBnd5} is derived by lower-bounding $\Phi_l$ by its value in the case $\bigl \lfloor \frac{d-1}{4} \bigr \rfloor =1$.
If we use a better lower bound of $\Phi_l(r_1,\dots,r_l)$ instead of that used in the proof of Theorem \ref{thm5}, an upper bound tighter than \eqref{eqBnd5} could be expected.
However, we can not get an explicit bound in this case  yet since the corresponding optimization problem is still very complicated.

\section{Construction Attaining The Upper Bound}
In this section, we give a new construction of binary linear LRCs.
The code has minimum distance $d\ge6$, and attains the upper bound \eqref{eqBnd5} in Theorem \ref{thm5}.

The construction relies on two matrices $A$ and $B$ defined as follows.
Suppose $s$ and $t$ are two positive integers such that $2t \mid s$ and $\frac{s}{2t} \ge 2$. Let $A$ be a binary matrix of size $2t \times 2^t$ such that any  $4$ columns of $A$ are linearly independent.
For $t \le 2$, $A$ can be  chosen as the identity matrix.
For $t \ge 3$, $A$ is a parity check matrix of a $[2^t,2^t-2t,5]$ binary code which can be constructed from nonprimitive cyclic codes of length $2^t+1$ (see e.g., \cite{chen1991construction}). We give a detailed construction of $A$ in Appendix \ref{appA}. Define $B$ to be a matrix whose columns are all nonzero $\frac{s}{2t}$-tuples from $\mathbb{F}_{2^{2t}}$ with first nonzero entry equal to $1$.
Then $B$ is actually a parity check matrix of a $2^{2t}$-ary Hamming code, and the size of $B$ is $\frac{s}{2t} \times \frac{2^s-1}{2^{2t}-1}$.

By fixing a basis of $\mathbb{F}_{2^{2t}}$ over $\mathbb{F}_2$, each vector in $\mathbb{F}_2^{2t}$ can be written as an element in $\mathbb{F}_{2^{2t}}$ and vice versa.
We denote by $a_1, \dots, a_{2^t} \in \mathbb{F}_{2^{2t}}$ the $2^t$ elements corresponding to the columns of $A$, and denote by a vector $\vec{\beta}_i\in\mathbb{F}_{2^{2t}}^{\,\frac{s}{2t}}$  the $i$th column of $B$ for $1 \le i \le \frac{2^s-1}{2^{2t}-1}$.
Then the binary linear LRC is constructed below.

\begin{construction}\label{cons1}
Define $\mathcal{C}$ to be a binary linear code with the parity check matrix
\begin{equation*}
  H = \begin{pmatrix}
    L_1 & L_2 & \dots & L_{l} \\
    H_1 & H_2 & \dots & H_{l}
  \end{pmatrix},
\end{equation*}
where $l = \frac{2^s-1}{2^{2t}-1}$, and
for $1 \le i \le l$, $L_i$ is an $l \times (2^t+1)$ matrix whose $i$-th row is the all-one vector and the other rows are all-zero vectors, $H_i$ is an $s \times (2^t+1)$ matrix over $\mathbb{F}_2$ whose columns are binary expansions of the vectors $\{\vec{0},a_1\vec{\beta}_i, a_2\vec{\beta}_i,\dots,a_{2^t}\vec{\beta}_i\}$.
\end{construction}

\begin{figure}[t]
\centering
\includegraphics[width=0.48\textwidth]{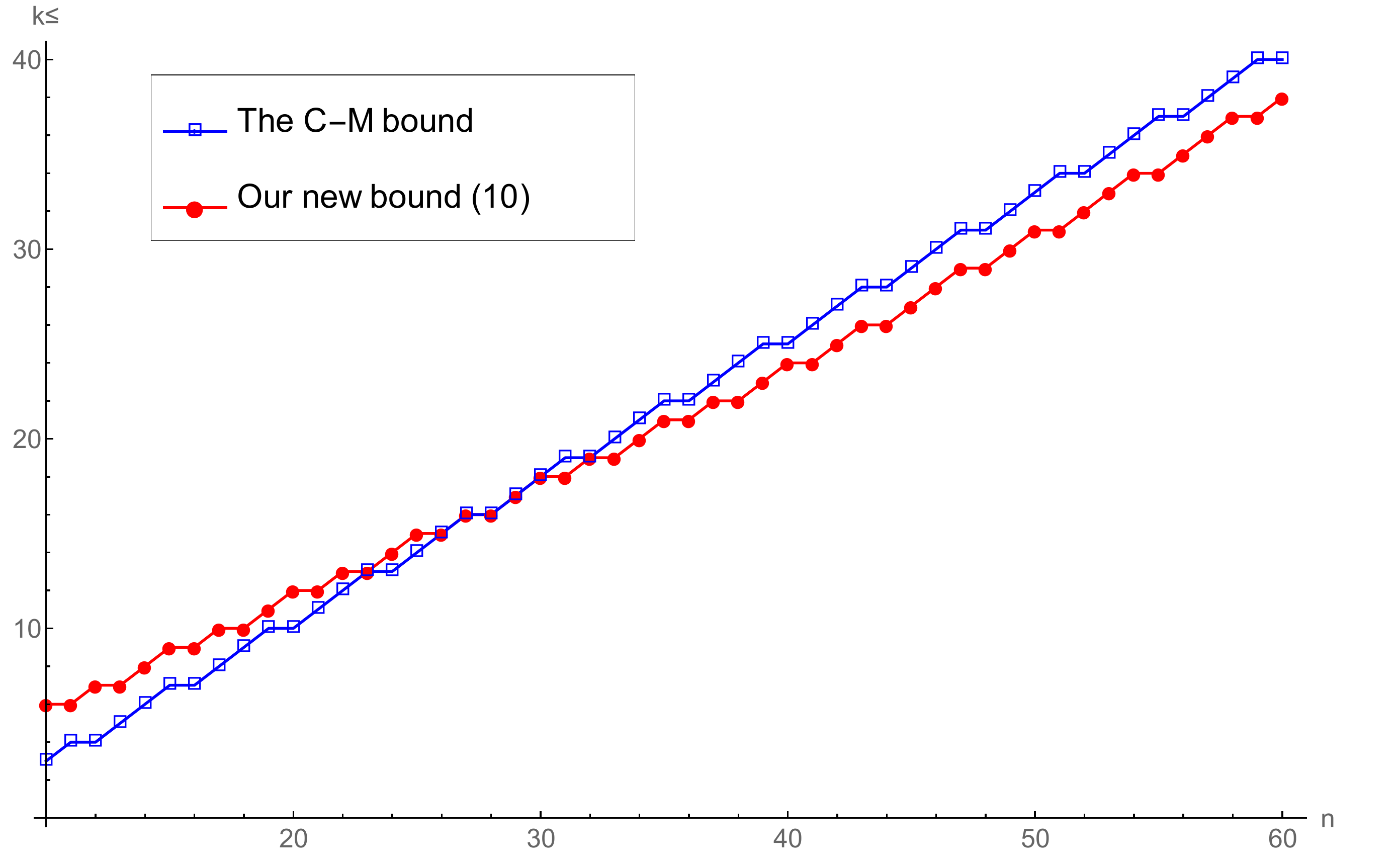}
\caption{A comparison of the C-M bound and the bound in Theorem \ref{thm5} for $r=3$, $d=5, 10 \le n \le 60$. }
\label{FigD5}
\end{figure}

\begin{figure}[t]
\centering
\includegraphics[width=0.48\textwidth]{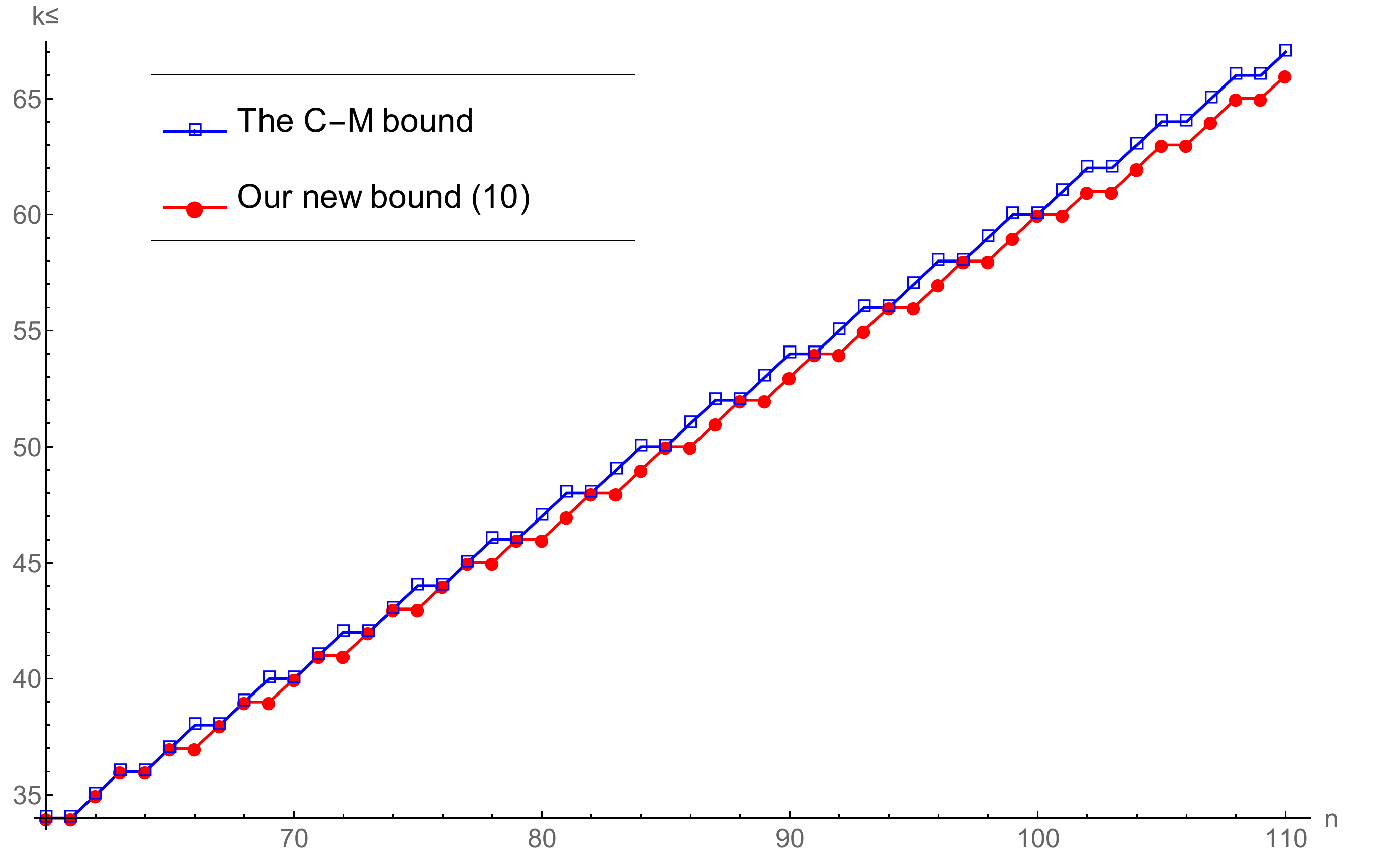}
\caption{A comparison of the C-M bound and the bound in Theorem \ref{thm5} for $r=2$, $d=8, 60 \le n \le 110$.}
\label{FigD8}
\end{figure}

\begin{example}
  Suppose $s=4$ and $t=1$.
  Then we can choose
  \begin{equation*}
    A = \begin{pmatrix}
      1 & 0 \\ 0 & 1
    \end{pmatrix} \in \mathbb{F}_2^{2 \times 2},
    \text{\hspace*{4pt}}
  B = \begin{pmatrix}
    1 & 1 & 1 & 1 & 0 \\
    \omega^2 & \omega & 1 & 0 & 1
  \end{pmatrix} \in \mathbb{F}_4^{2 \times 5},
  \end{equation*}
  where $\omega$ is a primitive element in $\mathbb{F}_4$ such that $\omega^2+\omega+1=0$.
  Fixing a basis $\{1, \omega\}$, the two columns of $A$ can be written as two elements in $\mathbb{F}_4$, i.e., $a_1 = (1, \omega) \cdot \bigl( \begin{smallmatrix} 1 \\ 0 \end{smallmatrix} \bigr)=1, a_2 = (1, \omega) \cdot \bigl( \begin{smallmatrix} 0 \\ 1 \end{smallmatrix} \bigr) =\omega$.
  Note that $\vec{\beta}_1 = (\begin{smallmatrix} 1 \\ \omega^2  \end{smallmatrix})$, then
  \begin{equation*}
    \alpha_1 \vec{\beta}_1 = \begin{pmatrix}       1  \\       \omega^2      \end{pmatrix},
    \text{\hspace*{4pt}}
    \alpha_2 \vec{\beta}_1 = \begin{pmatrix}       \omega \\      1     \end{pmatrix}.
  \end{equation*}
  By expanding $\{\vec{0}, \alpha_1 \vec{\beta}_1, \alpha_2 \vec{\beta}_1\} \subseteq \mathbb{F}_4^2$ into binary vectors with respect to the basis $\{1 ,\omega\}$, we get
  \begin{equation*}
    H_1 = \begin{pmatrix} 0 & 1 & 0 \\ 0 & 0 & 1 \\ 0 & 1 & 1 \\ 0 & 1 & 0 \end{pmatrix}.
  \end{equation*}
The other $H_i$'s can be computed similarly, so we have
  \begin{equation*}
    H = \begin{pmatrix}
      111 &        &       &       &       \\
            & 111 &       &       &       \\
            &       & 111 &       &       \\
            &       &       & 111 &       \\
            &       &       &       & 111 \\
      010 & 010 & 010 & 010 & 000 \\
      001 & 001 & 001 & 001 & 000  \\
      011 & 001 & 010 & 000 & 010  \\
      010 & 011 & 001 & 000 & 001  \\
    \end{pmatrix}.
  \end{equation*}
It can be verified that any  $5$ columns of $H$ are linearly independent.
So $H$ defines an $[n=15,k=6,d\ge6]$ binary LRC with locality $r=2$.
Substituting $n=15,d=6,r=2$ into the C-M bound~\eqref{eqCM} yields $k \le 6$, so this binary linear LRC is optimal with respect to the C-M bound.
\end{example}

\begin{theorem}
  The code $\mathcal{C}$ obtained from Construction \ref{cons1} is an binary linear LRC with $n = \frac{2^s-1}{2^t-1}, k \ge \frac{r n}{r+1} -s, d \ge 6$ and $r = 2^t$.
  Moreover, $\mathcal{C}$ attains the upper bound \eqref{eqBnd5} for all positive integers $s,t$ satisfying $2t \mid s$ and $\frac{s}{2t} \ge 2$ except the case $s=4, t=1$.
\end{theorem}
\begin{proof}
Since the values of $n,k,r$ can be determined easily, we focus on proving $d \ge 6$.
  Note that the sum of the first $l$ rows of $H$ is an all-one vector, so the minimum distance of $\mathcal{C}$ must be even.
  Therefore it suffices to show that $d \ge 5$.
  Suppose to the contrary that there exists a codeword $\vec{c} \in \mathcal{C}$ such that $H \vec{c}^\tau = 0$, $1 \!\le\! \wt(\vec{c}) \!\le\! 4$.
  Denote $\vec{c} \!=\! (\vec{c}_1,\dots, \vec{c}_l)$, where $\vec{c}_i \in \mathbb{F}_2^{2^t+1}$ for $i \in [l]$.
  It can be deduced from the definition of $L_i$ that $\wt(\vec{c}_i)$ is even, $\forall i \in [l]$.
  So there are at most two nonzero vectors in $\vec{c}_1,\dots, \vec{c}_l$.
  Without loss of generality, we suppose $\vec{c}_3 \!= \!\cdots \!=\! \vec{c}_l \!= \! 0$ and $1 \!\le\!\wt(\vec{c}_1, \vec{c}_2) \!\le\! 4$.
  Then by $H \vec{c}^\tau = 0$ it has $H_1 \vec{c}_1^\tau + H_2 \vec{c}_2^\tau = 0$.
  Denote $\vec{c}_1 = (x_0,x_1,\dots,x_{2^t})$ and $\vec{c}_2 = (y_0,y_1,\dots,y_{2^t})$, we have
  \begin{equation*}
    (x_1 a_1 + \dots + x_{2^t} a_{2^t}) \vec{\beta}_1 + (y_1 a_1 + \dots + y_{2^t} a_{2^t}) \vec{\beta}_2 =0.
  \end{equation*}
  Since $\vec{\beta}_1$ and $\vec{\beta}_2$ are linearly independent over $\mathbb{F}_{2^{2t}}$, it must has $(x_1 a_1 + \dots +  x_{2^t} a_{2^t}) = (y_1 a_1 + \dots + y_{2^t} a_{2^t})=0$, which contradicts to the fact that any $4$ out of $a_1, \dots, a_{2^t} \in \mathbb{F}_{2^{2t}}$ are linearly independent over $\mathbb{F}_2$.

  It remains to show $\mathcal{C}$ is optimal with respect to \eqref{eqBnd5}.
  Setting $n = \frac{2^s-1}{2^t-1}$ and $r = 2^t$, it has $(r+1) \mid n$, then it follows from \eqref{eqBnd5} that
  \begin{equation*}
    k \le \frac{rn}{r+1} - \min\{ \bigl\lceil \log_2(1+\frac{rn}{2}) \bigr\rceil , \bigl\lceil \frac{rn}{(r+1)(r+2)}\bigr\rceil \}.
  \end{equation*}
  We claim $\bigl\lceil \log_2(1+\frac{rn}{2}) \bigr\rceil = s$ and $\bigl\lceil \frac{rn}{(r+1)(r+2)} \bigr\rceil \ge s$ for all $s,t$ satisfying $2t \mid s$, $\frac{s}{2t} \ge 2$ except $s=4,t=1$.
  Then the claim implies that $k \le \frac{rn}{r+1} -s$, and therefore $\mathcal{C}$ is optimal.

  To show $\bigl\lceil \log_2(1+\frac{rn}{2}) \bigr\rceil = s$, note that
\begin{equation*}
  \log_2(1+\frac{rn}{2}) = \log_2(1+\frac{2^t}{2^t-1}\cdot\frac{2^s-1}{2}).
\end{equation*}
Since $1 < \frac{2^t}{2^t-1} \le 2$, we have
\begin{equation*}
  2^{s-1} < 1+\frac{2^t}{2^t-1}\cdot\frac{2^s-1}{2} \le 2^s.
\end{equation*}
It follows that $\bigl\lceil \log_2(1+\frac{rn}{2}) \bigr\rceil = s$.
It remains to show $\bigl\lceil \frac{rn}{(r+1)(r+2)} \bigr\rceil \ge s$.
When $t=1$, it has $s >4$ and
\begin{align*}
  \frac{rn}{(r+1)(r+2)} & = \frac{2^t}{(2^t+1)(2^t+2)} \cdot \frac{2^s-1}{2^t-1} \\
  & = \frac{1}{6} (2^s-1) \\
  & \ge s.
\end{align*}
When $t\ge2$, it has
\begin{align*}
  \frac{rn}{(r+1)(r+2)} & = \frac{2^t}{(2^t+1)(2^t+2)} \cdot \frac{2^s-1}{2^t-1} \\
  & \substack{(a) \\ \ge } \,\, \frac{2^t}{(2^t+1)(2^t+2)} \cdot \frac{1}{2^t-1} \cdot \frac{s(2^{4t}-1)}{4t} \\
  & = \frac{2^t}{2^t+2} \cdot \frac{2^{2t}+1}{4t} \cdot s \\
  & \substack{(b) \\ \ge } \,\, s,
\end{align*}
where (a) holds since $\frac{2^s-1}{s} \ge \frac{2^{4t}-1}{4t}$, which is a consequence of $s \ge 4t$, and (b) holds since $\frac{2^t}{2^t+2} \ge \frac{1}{2}$ and $2^{2t}+1 \ge 8t$ for $t \ge 2.$
\end{proof}

\section{Conclusions}
We introduce the concepts of $\mathcal{L}$-covers and $\mathcal{L}$-spaces for LRCs.
By using the sphere-packing bound in the $\mathcal{L}$-spaces, we derive new upper bounds on the dimension $k$ for binary linear LRCs. Two explicit bounds are given respectively for LRCs with and without the assumption of disjoint local repair groups. Comparing with previously known bounds for LRCs over particular finite fields, our bounds present an explicit form, generalize previous results to general cases, and outperform previous bounds in some cases. Moreover, a class of binary codes attaining our second bound are also designed. A further investigation into the $\mathcal{L}$-spaces are expected to bring more results in studying binary LRCs.

\newpage
\appendices

\section{}\label{appA}
Let $\beta$ be the primitive root of $x^{2^t+1}-1$, and let $M(x)$ denote the minimum polynimial of $\beta$, then $\deg(M(x)) = 2t$.
Define $\mathcal{A}$ to be the binary cyclic code of length $(2^t+1)$ which is generated $(x-1)M(x)$.
Then it can be checked that
\begin{equation*}
  \{\beta^i: i=-2,-1,0,1,2,\}
\end{equation*}
forms a subset of the roots of $(x-1)M(x)$.
It follows from the BCH bound that $\mathcal{A}$ is an $[2^t+1, 2^t-2t, \ge 6]$ binary linear code.
Then an $[2^t, 2^t-2t, \ge 5]_2$ punctured code can be obtained by deleting one coordinate of $\mathcal{A}$, and thus the matrix $A$ can be just chosen as the parity check matrix of the punctured code.

\end{document}